\documentclass[10pt]{article}
\usepackage{color, amssymb, amsthm, amsmath, amsfonts, ascmac, comment, enumerate}
\usepackage[dvipdfmx]{hyperref}
\usepackage{graphicx}
\usepackage{multirow}
\usepackage{tcolorbox}
\usepackage{xcolor}
\hypersetup{
  colorlinks, 
  citecolor=blue!20!black!30!green,
  linkcolor=red,
  urlcolor=blue}
\newtheorem{Thm}{Theorem}
\newtheorem*{Thm*}{Theorem}
\newtheorem{Prop}{Proposition}	
\newtheorem{Lem}{Lemma}
\newtheorem{Cor}{Corollary}

\newtheorem{Def}{Definition}
\newtheorem{Rem}{Remark}
\newtheorem*{Thm_non-iid}{\rm\bf Theorem~\ref{Thm_non-iid}}
\newtheorem*{Thm_Sec}{\rm\bf Theorem~\ref{Thm_Sec-order}}
\usepackage[top=30truemm,bottom=30truemm,left=20truemm,right=20truemm]{geometry}

\newcommand{\SPAN}{\mathrm{SPAN}}
\newcommand{\COST}{\mathrm{COST}}
\newcommand{\WORST}{\mathrm{WORST}}
\newcommand{\AVE}{\mathrm{AVE}}
\newcommand{\argmax}{\mathrm{argmax}}
\newcommand{\argmin}{\mathrm{argmin}}
\newcommand{\OPT}{\mathrm{opt}}
\newcommand{\BARE}{\bar{E}(\mathcal{S}_m)}
\newcommand{\UBARE}{\underline{E}(\mathcal{S}_m)}
\newcommand{\PRJ}{P_{\mathcal{J}^{n}}}
\newcommand{\SCHE}{(\phi_n, S_n)}
\newcommand{\SCHEP}{(\phi^\OPT_n, S^\OPT_n)}
\newcommand{\COSTOPT}{\COST\SCHEP}
\newcommand{\E}{\mathbf{E}}

\title{A probabilistic analysis on general probabilistic scheduling problems}
\author{Daiki Suruga}
\begin{document}

\maketitle
\begin{abstract}
The scheduling problem is a key class of optimization problems 
and has various kinds of applications both in practical and theoretical scenarios.
In the scheduling problem, probabilistic analysis is a basic tool for investigating performance of scheduling algorithms, and therefore 
has been carried out by plenty amount of prior works.
However, probabilistic analysis has several potential problems. For example, current research interest in the scheduling problem is limited to i.i.d. scenarios, due to its simplicity for analysis.
This paper provides a new framework for probabilistic analysis in the scheduling problem and aims to deal with such problems.
As a consequence, we obtain several theorems including a theoretical limit of the scheduling problem which can be applied to \emph{general, non-i.i.d. probability distributions}.
Several information theoretic techniques, such as \emph{information-spectrum method}, turned out to be useful to prove our results.
Since the scheduling problem has relations to many other research fields, our framework hopefully yields other interesting applications in the future.
\end{abstract}

\section{Introduction}
\subsection{Background}\label{Sec_Back}
\paragraph{Scheduling problems}
The scheduling problem is a class of optimization problems  in which we want to allocate a collection of jobs on machines appropriately in order to minimize (or maximize) a certain cost function.
In the most fundamental form, the subjective is to minimize the \emph{makespan}, which is the total completion time of all machines which need to process all of the allocated jobs.
The scheduling problem has been introduced explicitly over 70 years ago~\cite{PS09}, and plays a central role in various kinds of research fields since then.

In this paper we specifically focus on the \emph{uniform-machines scheduling problem}, (which includes the \emph{identical-machines scheduling problem}~\cite{McN59} as a special case),
which is one of the most well-investigated problems among many variants of scheduling problems.
In the uniform-machines scheduling problem, there are $m$ machines, and a job $j \in J$ takes processing time $p_{i,j} = p_j/v_i$ on  machine $i$, 
i.e., the processing time is determined as $p_j$ divided by the speed of the $i$-th machine $v_i~(1 \leq i \leq m)$ .
\par
The uniform-machines scheduling problem has so many applications in many kinds of practical scenarios.
This scheduling problem appears in real-world situations, such as production lines, university, hospitals, and computer systems~(see \cite{CS90, Mok01} for good surveys).
This is partially why the uniform machines scheduling problem has been paid much attention.
In addition to the importance in real-world situations, the uniform-machines scheduling problem also provides several interesting aspects in theoretical computer science.
For example, this problem is known to be NP-hard~\cite{Kar72, BKL77, GJ78, Har82, JK23} in general and is NP-complete when there are only two machines.
Since NP-hardness and NP-completeness are the central concepts in complexity theory,
the uniform-machines scheduling problem has been paid so much attention even in complexity theory.
As the NP-hardness implies, computing the exact solution is usually really hard, and therefore
this scheduling problem has often been examined by approximating the optimal solution~\cite{Gra66, Gra69, HS87, AAWY98, GG19}.
This implies that the scheduling problem has been paid attention also in approximation algorithm society, 
showing another interesting aspect in theoretical computer science.
Due to practical importance as well as theoretical interest,
the uniform-machines scheduling problem has become one of the core branches in several research fields, 
including operations research and computer science.

\paragraph{Probabilistic analysis}
Today, there are many algorithms~(often called as~\emph{heuristics} or \emph{policies}) for the scheduling problem, such as
the list scheduling (LS) algorithm and the largest processing time (LPT).
In such cases, to select an appropriate scheduling algorithm for one's purpose, the one need to evaluate each algorithm with an appropriately chosen measure of performance.
One such candidate of the measure is the worst-case performance evaluation: For an scheduling algorithm $A$, the worst-case performance of the algorithm $A$ is given by $\WORST(A):=\max_{j^n \in J^n} A(j^n)$ where $A(j^n)$ denotes the makespan of the algorithm with input $j^n = (j_1, \ldots, j_n)$, which is a list of $n$ jobs.
This measure of performance gives a theoretical guarantee that any scheduling instance of the algorithm $A$ always has the makespan less than (or equal to) the quantity $\WORST(A)$.
This is one benefit of the worst-case analysis.
However the worst-case analysis has the crucial weak point: The makespan $A(j^n)$ of an input $j^n$ is often significantly shorter than the value $\WORST(A)$.
In other words, the hardest list of jobs:~$j^n_\mathrm{HARD} := \argmax A(j^n)$ may be really unlikely to happen in practice.
Today, there are several other measures of performance applied to scheduling algorithm, to overcome the weak point of the worst-case analysis.
\par
The average-case analysis, which is the other well-known method of the evaluation, can deal with the weak point of the worst-case analysis.
In the average-case analysis, a distribution $\Pr_{J^n}$ on $J^n$ is defined appropriately in order to reflect a practical situation, 
and performance of an algorithm $A$ is given as the expectation $\AVE(A, \Pr_{J^n}) := \E_{j^n \sim J^n}[A(j^n)]$ under the distribution $\Pr_{J^n}$.
As is easily seen, this analysis overcomes the weak point of the worst-case analysis, since jobs which occur with small probability affect the value $\AVE(A, \Pr_{J^n})$ not so much, even if makespans of the jobs are really large.
Therefore, the average-case analysis is usually more appropriate for practical situations and has been paid attention in the scheduling problem.
\par
As mentioned above, the average-case analysis has been paid attention over decades, 
and therefore there is a certain number of works that analyze average-case behaviors of scheduling algorithms~(for example, \cite{PR96, CCG+00, CFGK91, CJLS93, CFW96, BBC18, BCEH21}), 
including analysis on an optimal algorithm that always outputs an optimal schedule, a schedule having the shortest makespan.
 References~\cite{Kno81, Lue82, CFGK91, CCG+02, PR96} treat the average-case behaviors of \emph{optimal algorithms} for the identical-machines scheduling problem or related problems.
For example, 
Ref.~\cite{PR96} characterized the constant $\theta := \lim_{n \to \infty}\E[\mathrm{OPT}_n(j^n)]/n$ when a list of jobs $j^n$ is $n$-i.i.d. distributed where $\mathrm{OPT}_n$ denotes an optimal algorithm.
On the other hand, some prior works such as References~\cite{FK87, CJLS93, CFW96, Li99, BCEH21} treat the average-case behaviors of several \emph{practical algorithms}.
In the case of practical algorithms, 
for example, 
Ref.~\cite{FK87} showed that the LPT algorithm asymptotically becomes the optimal one. 
That is, the makespan by the LPT algorithm (almost surely) converges to that of the optimal algorithm as $n \to \infty$, when a list of jobs $j^n$ is $n$-i.i.d. distributed.
For recent results on the LPT algorithm, see~Ref.~\cite{BCEH21} for example.
For its practical importance as well as theoretical interests, the average-case analysis on the scheduling problem has been investigated in many prior works.
\par
Considering the usefulness of the average-case analysis, it may seem that the average-case analysis is the best way to evaluate performance of scheduling algorithms, and has no weak point at all.
However, the average-case analysis has at least two following issues.
First, the average-case analysis usually need to assume i.i.d. conditions on distributions $\Pr_{J^n}$, which should reflect a practical situation appropriately.
Indeed, to the best of our knowledge, almost all the previous works\footnote{Several exceptions~(e.g., \cite{EMP01, LNCH04} in real-time analysis literature) can be found, which are not closely related to this paper.} assume the i.i.d. conditions. However, in practice this is too restrictive; 
there are many practical situations where distributions of jobs are correlated.
Jobs may form a Markov chain.
Second, in contrast to the worst-case analysis, the expectation $\E_{j^n \sim J^n}[A(j^n)]$ does not guarantee the worst-case performance of an algorithm $A$. 
This means, the makespan $A(j^n)$ of some jobs $j^n$ may become really long even if the average $\E_{j^n \sim J^n}[A(j^n)]$ is relatively short.
Since the time for process is usually limited in practice, the schedule by the algorithm $A$ will not always work correctly and may cause some serious problems due to this unexpectedness.
These are the two issues of the average-case analysis, and 
one of our main focus in this paper is to deal with the two above problems by appropriately re-defining a measure of performance.

\subsection{Our contributions}
As mentioned in Section~\ref{Sec_Back}, one of the issues regarding the average-case analysis
is that current research focus only on the i.i.d. distributions, even though non-i.i.d. distributions will occur in many practical cases.
Therefore, in this paper, we focus on the asymptotic behavior of optimal scheduling algorithm under \emph{non-i.i.d., general probability distributions}.
\par
Our results can be divided into three parts.
\begin{description}
\item[\bf First Result:] The asymptotic limit of the expectation under general probability distributions in Section~\ref{Sec_average},
\item[\bf Second Result (main result):] The asymptotic limit of a new performance measure  under general probability distributions in Section~\ref{Sec_Char-of-opt},
\item[\bf Third Result:] Precise analysis on a new performance measure in case of the i.i.d. distributions in Section~\ref{Sec_Sec-order_asymptotics}.
\end{description}
In First Result,  we characterize the asymptotic limit of the expectation of the optimal algorithm under general probability distributions.
The main result of this paper is Second Result. Second Result aims to resolve both of the two issues mentioned in Section~\ref{Sec_Back}.
The two issues say that both the worst-case analysis and the average-case analysis may not appropriate for practical scenarios, 
and hence another performance measure needs to be introduced.
The new performance measure is then suitably defined to resolve such issues. Even though the new performance measure had not been introduced in scheduling theory,
similar measures have played as natural measures of performance in several research fields such as in Information theory.
Second Result is devoted to the asymptotic characterization of the new performance measure in general probability distributions, which is successfully done by applying the well-established technique called \emph{Information-spectrum method}.
Third Result provides more precise analysis, without relying on the asymptotic limit,  on the new measure of performance in case of i.i.d. distributions.
Let us explain each of the results separately in detail.

\subsubsection{First result}
Our first result examines the asymptotic limit of the expectation of the optimal algorithm:
\[
\limsup_{n \to \infty} \frac{1}{n}\E_{J^n}[\SPAN(\phi^\mathrm{opt}_n, J^n)],
\]
where $\SPAN(\phi_n^\mathrm{opt}, J^n)$ is the makespan of the optimal algorithm $\phi_n^\mathrm{opt}$ with input $J^n$.
Note that we need to use ``$\limsup$'' instead of ``$\lim$'' because of the generality of probability distributions.
In Section~\ref{Sec_average}, we characterize this quantity by another quantity which is expressed as the limit of the expectation of random variables.
(See Equation~\eqref{eq_Thm_expectation} for the explicit expression.)
Even though our characterization may seem difficult to compute, it becomes simpler in many important cases by applying techniques from probability theory.
For example,
a simplified proposition, 
Proposition~\ref{Prop_expectation_iid}, is obtained from the law of large numbers applied to i.i.d. scenarios,
and another simplified proposition, Proposition~\ref{Prop_expectation_Markov}, is obtained from the ergodic theorem applied to markov distributions.
Note that as a prior work, Ref.~\cite{PR96} shows the same result in case of the i.i.d. distributions. 
Compared to Ref.~\cite{PR96}, our result is proved in an arguably simpler way because only combinatorial approach is used to show Proposition~\ref{Prop_expectation_iid},
whereas Ref.~\cite{PR96} uses several other non-trivial techniques such as linear programming and its Lagrange relaxation.

\subsubsection{Second result} 
As mentioned, a new measure of performance is introduced in order to evaluate practical performance appropriately.
In short, our new measure discards a set $S_n$ of jobs that are \emph{unlikely to occur},
and evaluates the worst-case performance over the set $S^c_n$ of jobs which are likely to occur.\footnote{$A^c$ denotes the complement set of $A$.}
Formally, the new measure is defined as follows.
\begin{Def}
Let $S_n \subseteq J^n$ be a set of jobs unlikely to occur, i.e., $\Pr(S_n)$ is close to zero.
For a scheduling algorithm $\phi_n$ on $J^n$,
the new measure $\COST(\phi_n, S_n)$ is defined as
\begin{equation*}
\COST(\phi_n, S_n)
= \max_{j^n \in S_n^c} \phi_n(j^n)
\end{equation*}
where $\phi_n(j^n)$ is the makespan by the algorithm $\phi_n$ with an input $j^n$.
\end{Def}
This measure indeed captures the worst-case performance in a practical scenario,
while the average-case analysis fails to do so.
To see this, consider a scenario where a scheduling algorithm $\phi_n$ processes an input $j^n \in J^n$ only when the input belongs to the set $S^c_n$; otherwise the jobs $j^n$ are discarded and not scheduled to machines.
In this scenario, the value $\COST(\phi_n, S_n)$ obviously corresponds to the worst-case makespan produced by the algorithm $\phi_n$.
In our main contributions, Theorem~\ref{Thm_non-iid} and Theorem~\ref{Thm_Sec-order}, we successfully analyze fundamental behaviors of this new measure of performance.
In Theorem~\ref{Thm_non-iid}, we evaluate the limit: $\limsup_{n \to \infty} \frac{1}{n}\COST(\phi_n, S_n)$ for efficient algorithms\footnote{We often simply call a pair $(\phi_n, S_n)$ as an \emph{algorithm}.} $(\phi_n, S_n)$ 
whose discarding probability asymptotically vanishes, i.e.,  $\Pr_{J^n}(S_n) \to 0$ as $n \to \infty$.
More precisely, in Theorem~\ref{Thm_non-iid} we evaluate the smallest possible value of $\limsup_{n \to \infty} \frac{1}{n}\COST(\phi_n, S_n)$ when $\Pr_{J^n}(S_n)$ must satisfy $\Pr_{J^n}(S_n) \to 0$.
To state the result formally, we introduce the notion of \emph{rate} as follows:
\begin{Def}
The rate $R > 0$ is achievable if and only if there is  a (sequence of) scheduling algorithm(s): $\{(\phi_n, S_n)\}_n$
such that $\lim_{n \to \infty} \Pr_{J^n}(S_n) = 0$ and $\limsup_{n \to \infty}\frac{1}{n}\COST(\phi_n, S_n) \leq R$ holds.
\end{Def}
Our goal is then to evaluate the infimum of the achievable rate: 
$\inf\{R > 0 \mid \text{The rate $R$ is achievable.} \}$,
and we characterize this quantity as follows:
\begin{Thm}\label{Thm_non-iid}
For any $m$-uniform machines scheduling problem, 
\begin{equation*}
\inf\{R > 0 \mid \text{The rate $R$ is achievable.} \} = \bar{E}
\end{equation*}
 holds
where
$\bar{E} := \inf\{\alpha \geq 0 \mid \lim_{n \to \infty} P_{\mathcal{J}^n}\{j^n \mid \frac{1}{n \cdot v_\mathrm{sum}} \sum_{i \leq n}T(j_i) > \alpha\} = 0\}$,
$v_\mathrm{sum}$ is the sum of speeds of all machines and $T(j_i)$ denotes the processing time of a job $j_i$ at unit speed.
\end{Thm}
Let us mention several remarks regarding Theorem~\ref{Thm_non-iid}.
First and foremost, Theorem~\ref{Thm_non-iid} assumes almost nothing,  and thus can be applied to any scheduling problem with any kind of distribution.
In particular, we can apply Theorem~\ref{Thm_non-iid} to scheduling problems with non-i.i.d. distributions such as Markov processes, ergodic processes.
In the case of Markov processes, for example, the value $\bar{E}$ turns out to be equal to $\frac{1}{v_\mathrm{sum}}\E_J[T(J)]$ where $\E_J[T(J)]$ denotes the expectation of the processing time $T(J)$ with respect to its stationary distribution.
Second, considering that the discarding probability $\Pr_{J^n}(S_n)$ must need to vanish asymptotically, 
one may expect that there is asymptotically no difference between the worst-case measure $\frac{1}{n}\WORST[\mathrm{OPT}_n]$ and our evaluation $\bar{E}$.
However, this is not true. Even in a simple i.i.d. scenario, we see that $\frac{1}{n}\WORST[\mathrm{OPT}_n] \to \frac{1}{v_\mathrm{sum}}T_{\mathrm{max}}$ where $T_\mathrm{max} := \max_j T(j)$, but $\bar{E} = \frac{1}{v_\mathrm{sum}}\E[T(J)]$.
This means that there is a certain amount of savings by admitting the negligible probability of discarding.
Lastly, for several important cases such as i.i.d. scenarios, the quantity $\bar{E}$ and the average case evaluation $\lim_{n \to \infty} \E_J[\mathrm{OPT}_n(J)]/n$ coincide.
\subsubsection{Third result}

Even though our first result successfully characterizes the asymptotic rate of any scheduling problem with any distribution of jobs,
the result only gives the optimal rate in the asymptotic setting.
Needless to say, there are many practical situations where one may want to estimate the optimal value of $\COST(\phi_n, S_n)$ for finite number of jobs, e.g., $n = 1000$.
To deal with this problem, in our third result, we employ more precise analysis on the value $\COST(\phi^\OPT_n, S^\OPT_n)$ 
where $(\phi_n^\OPT, S_n^\OPT)$ is an optimal scheduling (with discarding probability $\varepsilon$) defined as $(\phi_n^\OPT, S_n^\OPT) = \argmin_{\Pr(S_n) \leq \varepsilon} \COST(\phi_n, S_n)$,
and characterize this value accurately for any finite $n$, up to an additive constant factor.
Our third result can be applied to any kind of i.i.d. scheduling problem.
The formal statement of our third result is as follows.
\begin{Thm}\label{Thm_Sec-order}
For any i.i.d. scheduling problem, 
\begin{equation*}
\COST\SCHEP
= \frac{n}{v_\mathrm{sum}}\E[T(J)] - \frac{\sqrt{V(T(J))n}}{v_\mathrm{sum}}\Phi^{-1}(\varepsilon) + \Theta(1).
\end{equation*}
holds where $\Phi(x)$ is the standard Gaussian distribution and $V(T(J)) := \E_J[(T(J) - \E[T(J)])^2]$.
\end{Thm}

\subsubsection{Our framework and Information theory}
Even though our results may seem no obvious relation to Information theory, they actually have something in common in spirit.
This section briefly describes the underlying connection between our framework and fixed length noiseless coding in Information theory.
\par
In fixed length noiseless coding scheme, a coding is defined as a pair $\phi_n = (\phi_n^\mathrm{Enc},\phi_n^\mathrm{Dec})$
where an encoder $\phi_n^\mathrm{Enc}: \mathcal{X}^n \to \{1, \ldots, M_n\}$  maps a list of $n$ symbols $(x_1, \ldots, x_n) \in \mathcal{X}^n$
to an integer $\phi_n^\mathrm{Enc}(x_1, \ldots, x_n) \in \{1, \ldots, M_n\}$, and a decoder $\phi_n^\mathrm{Dec}$ does the same reversely.
In this scheme, $n$ symbols in $\mathcal{X}^n$ are randomly generated, and the error probability is introduced as the probability of the event $S_n := \{x \in \mathcal{X}^n \mid \phi_n^\mathrm{Dec} \circ \phi_n^\mathrm{Enc}(x) \neq x\}$.

The seminal paper~\cite{Sha48} shows that the rate ``$\lim \frac{1}{n} \log M_n$'' of an optimal coding subject to $\Pr(S_n) \to 0$ satisfies
\begin{equation}\label{eq_Sha48}
\lim_{n \to \infty} \frac{1}{n}\log M_n = H(X)
\end{equation}
when the symbols are i.i.d.
($H(X)$ is the entropy of the randomly generated symbol on $\mathcal{X}$.)
This result characterizes the \emph{asymptotic} behavior of the optimal rate for \emph{i.i.d. symbols}.
Since the original result is limited to i.i.d. scenarios, Ref.~\cite{Han03} established a novel method, \emph{Information-Spectrum method}, to deal with non i.i.d. scenarios
and showed that the optimal rate is, instead of the entropy $H(X)$, characterized as
\begin{equation}\label{eq_Han03}
\bar{H}(X) := \inf\left\{\alpha \geq 0 \mid \lim_{n \to \infty}\Pr\{x^n \mid \frac{1}{n}\log\frac{1}{\Pr_{\mathcal{X}^n}(x^n)} > \alpha\} = 0\right\}
\end{equation}
for general probability distributions.
There is another way of generalizing the original result called \emph{second-order asymptotics} (e.g., \cite{Str62, Tan14}) which investigates non-asymptotic aspects of fixed length coding.
Without relying on the limit, Ref.~\cite{Str62} characterizes the optimal rate as
\begin{equation*}
\log M_n = n \cdot H(X) + \sqrt{nV(X)} \Phi^{-1}(\varepsilon) - \frac{1}{2}\log n + \Theta(1)
\end{equation*}
for i.i.d. symbols. ($V(X)$ is the variance of $\log \frac{1}{\Pr(X)}$ and $\varepsilon := \Pr(S_n)$.)

\par
There are several similarities observed between fixed length coding scheme and our framework for scheduling problem.
Job are randomly generated, the set $S_n$ represents the event of error, and, the notion of rate $\lim_{n\to \infty}\frac{1}{n} \COST(\phi_n, S_n)$ is defined in a similar fashion
to fixed length coding.
Due to such similarities, this paper successfully adapts the information-spectrum method to our framework and obtains Theorem~\ref{Thm_non-iid} which deals with non-i.i.d. scenarios.
This paper also adapts the second order asymptotic method and obtains Theorem~\ref{Thm_Sec-order}.
(We also obtain a result corresponding to the strong converse theorem, which is left to Appendix~\ref{Sec_converse}.)
Even though several modifications are necessary to adapt these techniques to our framework, our results certainly have some connection to Information theory.

\section{Preliminaries}
Throughout this paper, we focus on the $m$-machine scheduling problem which is defined as a tuple $\mathcal{S}_m = (\mathcal{J}, T, \mathcal{V}, \mathcal{P})$ where
\begin{itemize}
\item $\mathcal{J} = \{j_1, \ldots, j_n\}$ denotes a set of $n$-jobs,
\item $T: \mathcal{J} \to \mathbb{N}$ denotes a function whose output $T(j_i)$ represents processing time of a job $j_i$ at unit speed,
\item $\mathcal{V} = \{v_1, \ldots, v_m\} \subset \mathbb{R}_{> 0}$ denotes a set of speeds of machines, where $v_i$ is the speed of machine $i$.\footnote{A job $j_i$ is processed by a machine $k$ with time $T(j_i)/v_k$.}
\item $\mathcal{P} = \{P_{\mathcal{J}^n}\}_n$ denotes a sequence of probability distributions where each $P_{\mathcal{J}^n}$ is a probability distribution on the $n$-product set $\mathcal{J}^n$.
\end{itemize}
For a scheduling problem $\mathcal{S}_m$ and each $n \in \mathbb{N}$, a scheduling $(\phi_n, S_n)$ is defined as 
\begin{itemize}
\item a scheduler $\phi_n$ which takes $j^n \in \mathcal{J}^n$ as input and outputs an allocation of the jobs $j^n$ to $m$ machines,
\item a subset of jobs $S_n \subseteq \mathcal{J}^n$ corresponding to jobs which are discarded, not executed or just out of consideration.
Without loss of generality, $S_n$ is assumed to satisfy $S_n \neq \mathcal{J}^n$.
\end{itemize}
Note that this definition is reduced to the ordinary one when $S_n$ is taken to be empty.
\par
Let $\SPAN(\phi_n, j^n)$ denote the makespan of a scheduler $\phi_n$ on a job $j^n \in \mathcal{J}^n$.
The worst-case cost $\COST(\phi_n, S_n)$ of the scheduling $(\phi_n, S_n)$ is the maximum of the makespan $\SPAN(\phi_n, j^n)$ over all jobs $j^n \in \mathcal{J}^n \setminus S_n$.
\par
Let us next introduce several notations to state some of our results in a simpler way.

\begin{Def}
$\bar{E}(\mathcal{S}_m) := \inf\{\alpha \geq 0 \mid \lim_{n \to \infty} P_{\mathcal{J}^n}\{j^n \mid \frac{1}{n \cdot v_{\mathrm{sum}}}T_n(j^n) > \alpha\} = 0\}$
where $T_n(j^n) := \sum_{i_\leq n} T(j_i)$ is the sum of processing time over all jobs in the list $j^n$ and $v_\mathrm{sum} := \sum_{i \leq m} v_i$.
\end{Def}

To evaluate how $\frac{1}{n} \COST(\phi_n, S_n)$ grows, we define the notion of rate $R$ as follows.
\begin{Def}
A rate $R \in \mathbb{R}_{\geq 0}$ is achievable if and only if there is a sequence of scheduling $\{(\phi_n, S_n)\}_n$ such that
$\lim_{n \to \infty} P_{\mathcal{J}^n}(S_n) = 0$ and 
$\limsup_{n \to \infty} \frac{1}{n} \COST(\phi_n, S_n)\leq R$ hold.
\end{Def}
We then observe that a scheduling with smaller rate is regarded as more efficient,
and therefore, we are interested in the minimization of the rate. That is, 
\begin{Def}
$R(\mathcal{S}_m) := \inf\{R \mid \text{$R$ is achievable.}\}$
\end{Def}

\section{Lemmas}
Here we prove several basic or fundamental lemmas required to show our main theorems.
\begin{Lem}
For an optimal scheduling $(\phi^\mathrm{opt}_n, S_n)$, i.e., $\SPAN(\phi^\mathrm{opt}_n, j^n) = \min_{\phi_n'}\SPAN(\phi_n', j^n)$ for any $j^n \in \mathcal{J}^n$,
\begin{equation*}
\frac{1}{m}\left(1 - \frac{m}{n}\right) \frac{T_\mathrm{min}}{v_\mathrm{max}} 
\leq \frac{1}{n}\COST(\phi^\mathrm{opt}_n, S_n) 
\leq \frac{1}{m}\left(1 + \frac{m}{n}\right) \frac{T_\mathrm{max}}{v_\mathrm{min}}
\end{equation*}
holds where $T_\mathrm{min} := \min_{j \in \mathcal{J}} T(j)$, $T_\mathrm{max} := \max_{j \in \mathcal{J}} T(j)$,
$v_\mathrm{max} := \max_{i \leq m} v_m$ and $v_\mathrm{min} := \min_{i \leq m} v_m$.
\end{Lem}
\begin{proof}
Suppose $j^n = (j_1, \ldots, j_n)$ satisfies $T(j_i) = T_\mathrm{min}$ for any $i \leq n$.
This list of jobs $j^n$ trivially yields the minimum makespan among all jobs in $\mathcal{J}^n$ when a scheduling is optimal.
We now evaluate $\SPAN(\phi^\mathrm{opt}_n, j^n)$. 
First, consider a scenario where each machine has the speed $v_\mathrm{max}$, faster than the actual one.
Then by an analogue of the pigeon hole principle, the optimal makespan becomes $(1_{\{c' \neq 0\}} + n') T_\mathrm{min}/v_\mathrm{max}$
where $n'$ and $c$ are two natural numbers satisfying $n = n' m + c~(0 \leq c < m)$.
Since the actual machines are slower, we have
\begin{equation*}
\SPAN(\phi^\mathrm{opt}_n, j^n) \geq (1_{\{c' \neq 0\}} + n') \frac{T_\mathrm{min}}{v_\mathrm{max}}
\end{equation*}
which leads to 
\begin{equation*}
\SPAN(\phi^\mathrm{opt}_n, j^n) \geq (1_{\{c' \neq 0\}} + n') \frac{T_\mathrm{min}}{v_\mathrm{max}}
\geq n' \frac{T_\mathrm{min}}{v_\mathrm{max}}
= \frac{n - c}{m}\frac{T_\mathrm{min}}{v_\mathrm{max}}
\geq \frac{n - m}{m}\frac{T_\mathrm{min}}{v_\mathrm{max}}
.
\end{equation*}
By dividing the both sides with $n$, we get 
\[
\frac{1}{m}\left(1 - \frac{m}{n}\right) \frac{T_\mathrm{min}}{v_\mathrm{max}} \leq \frac{1}{n}\COST(\phi^\mathrm{opt}_n, S_n).
\]
\par
On the other hand, suppose $j^n = (j_1, \ldots, j_n)$ satisfies $T(j_i) = T_\mathrm{max}$ for any $i \leq n$.
Then, in a similar manner, we have
\begin{equation*}
\SPAN(\phi^\mathrm{opt}_n, j^n) = (1_{\{c' \neq 0\}} + n') \frac{T_\mathrm{max}}{v_\mathrm{min}}
\leq (1 + n') \frac{T_\mathrm{max}}{v_\mathrm{min}}
= \frac{m + n - c}{m}\frac{T_\mathrm{max}}{v_\mathrm{min}}
\leq \frac{n + m}{m}\frac{T_\mathrm{max}}{v_\mathrm{min}}.
\end{equation*}
This yields 
\[
\frac{1}{n}\COST(\phi^\mathrm{opt}_n, S_n) 
\leq \frac{1}{m}\left(1 + \frac{m}{n}\right) \frac{T_\mathrm{max}}{v_\mathrm{min}}
\]
which completes proof.
\end{proof}
By defining a sequence $\{(\phi_n, S_n)\}$ as $\phi_n = \phi_n^\mathrm{opt}$ and $S_n = \emptyset$, and by taking the limit $\lim_{n \to \infty}$,
we obtain the following corollary. 
\begin{Cor}
The optimal rate $R(\mathcal{S}_m)$ belongs to the interval $[T_\mathrm{min}/(m \cdot v_\mathrm{max}), T_\mathrm{max}/(m\cdot v_\mathrm{min})]$.
\end{Cor}

\begin{Lem}\label{Lem_sche_upper-bound}
For any $n$, there is a scheduler $\phi_n$ such that, for any list of jobs $j^n = (j_1, \ldots, j_n)$,
\begin{equation}\label{eq_lem_sche_upper-bound}
 \SPAN(\phi_n, j^n) \leq \frac{1}{v_\mathrm{sum}} T_n(j^n) + \frac{T_\mathrm{max}}{v_\mathrm{min}}
\end{equation}
holds.
\end{Lem}

\begin{proof}
Take the optimal scheduler $\phi_n^\OPT$, i.e., $\SPAN(\phi^\textrm{opt}_n, j^n) = \min_{\phi_n'}\SPAN(\phi_n', j^n)$ for any $j^n \in \mathcal{J}^n$.
We prove Equation~(\ref{eq_lem_sche_upper-bound}) by assuming the opposite inequality $\SPAN(\phi^\OPT_n, j^n) > \frac{1}{v_\mathrm{sum}} T_n(j^n) + \frac{T_\mathrm{max}}{v_\mathrm{sum}}$ and giving a contradiction.
\par
Suppose $\SPAN(\phi^\OPT_n, j^n) > \frac{1}{v_\mathrm{sum}} T_n(j^n) + \frac{T_\mathrm{max}}{v_\mathrm{min}}$ holds. 
In this case, there needs to be one or more machines which finish all of the allocated jobs before the time $T_n(j^n)/v_\mathrm{sum}$ with the scheduling $\phi^\OPT_n$ on the jobs $j^n$.
To see this, let us calculate the maximum amount of jobs that the machines are able to process within the time $T_n(j^n)/v_\mathrm{sum}$.
The machine 1 is able to process at most $v_1 \cdot T_n(j^n)/v_\mathrm{sum}$ amount of jobs within the time $T_n(j^n)/v_\mathrm{sum}$, which is followed by the definition of the speed.
Applying this idea for all $i\leq m$, we see that, for any $i \leq m$ ,
each  machine $i$ is able to process at most $v_i \cdot T_n(j^n)/v_\mathrm{sum}$ amount of jobs within the time $T_n(j^n)/v_\mathrm{sum}$.
This implies the amount of jobs processed by all of the machines with the time $T_n(j^n)/v_\mathrm{sum}$ is 
\begin{equation}\label{eq_sum_jobs}
v_1 \frac{T_n(j^n)}{v_\mathrm{sum}}
+
v_2 \frac{T_n(j^n)}{v_\mathrm{sum}}
+
\cdots
+
v_m \frac{T_n(j^n)}{v_\mathrm{sum}}
=
T_n(j^n)
\end{equation}
which equals to the entire amount of jobs the machines needed to process, if the jobs are fully assigned.
By the assumption of $\SPAN(\phi^\OPT_n, j^n)$ which is strictly larger than the time $T_n(j^n)/v_\mathrm{sum}$, 
we therefore see that there is one or more machines which finish all of the allocated jobs before the time $T_n(j^n)/v_\mathrm{sum}$.
We label one such machines by $i_0 \in \{1, \ldots, m\}$.

\par
The assumption, $\SPAN(\phi^\OPT_n, j^n) > \frac{1}{v_\mathrm{sum}} T_n(j^n) + \frac{T_\mathrm{max}}{v_\mathrm{min}}$,
also means that there exists one or more jobs whose processing starts later than $\frac{1}{v_\mathrm{sum}} T_n(j^n)$ and finishes exactly at $\SPAN(\phi^\OPT_n, j^n)$
since $\SPAN(\phi^\OPT_n, j^n) - \frac{1}{v_\mathrm{sum}} T_n(j^n)> T_\mathrm{max}/v_\mathrm{min}$ holds and any job processed by any machine requires processing time less than or equal to $T_\mathrm{max}/v_\mathrm{min}$.
Then, the job can be re-assigned to the machine $i_0$, which finishes all of it's assigned jobs earlier than $\frac{1}{v_\mathrm{sum}} T_n(j^n)$.
After re-assigning the job, we recursively re-assign other jobs whose processing finishes exactly at $\SPAN(\phi^\OPT_n, j^n)$ in the same manner until all the jobs have been re-assigned to such machines.
This new scheduling has the makespan strictly less than $\SPAN(\phi^\OPT_n, j^n)$. This contradicts the fact that $\phi^\OPT_n$ is the optimal scheduling and therefore obtain
\begin{equation*}
 \SPAN(\phi^\OPT_n, j^n) \leq \frac{1}{v_\mathrm{sum}} T_n(j^n) + \frac{T_\mathrm{max}}{v_\mathrm{min}}
\end{equation*}
which completes the proof.
\end{proof}
 
\begin{Lem}\label{Lem_sche_lower-bound}
For any $n$, any scheduler $\phi_n$ and any list of jobs $j^n = (j_1, \ldots, j_n)$,
\begin{equation}\label{eq_lem_sche_lower-bound}
 \frac{1}{v_\mathrm{sum}}T_n(j^n) \leq \SPAN(\phi_n, j^n)
\end{equation}
holds.
\end{Lem}
\begin{proof}
Take the optimal scheduler $\phi_n^\OPT$, i.e., $\SPAN(\phi^\textrm{opt}_n, j^n) = \min_{\phi_n'}\SPAN(\phi_n', j^n)$ for any $j^n \in \mathcal{J}^n$.
We then only need to show $\frac{1}{v_\mathrm{sum}}T_n(j^n) \leq \SPAN(\phi^\OPT_n, j^n)$.
Suppose the opposite direction $\frac{1}{v_\mathrm{sum}}T_n(j^n) > \SPAN(\phi^\OPT_n, j^n)$ holds.
This means that every machine completes the assigned jobs earlier than the time $\frac{1}{v_\mathrm{sum}}T_n(j^n)$.
Therefore, the sum of the amount of jobs processed by all machines is less than $T_n(j^n)$ which follows from the same argument as in Equation~\ref{eq_sum_jobs}.
This leads to a contradiction since the sum of the amount of jobs for all machines must be equal to $T_n(j^n)$, as desired.
\end{proof}

\section{The average case analysis}\label{Sec_average}

From Lemma~\ref{Lem_sche_upper-bound} and Lemma~\ref{Lem_sche_lower-bound}, we have
\begin{equation*}
\frac{1}{v_\mathrm{sum}}T_n(j^n) \leq \SPAN(\phi_n, j^n) \leq \frac{1}{v_\mathrm{sum}}T_n(j^n) + \frac{T_\mathrm{max}}{v_\mathrm{min}}.
\end{equation*}
Divide each of them by $n$ and take $\limsup_{n \to \infty}$ yields the following equation~\eqref{eq_Thm_expectation}.

\begin{equation}\label{eq_Thm_expectation}
\limsup_{n \to \infty} \frac{1}{n}\E_{J^n}[\SPAN(\phi^\mathrm{opt}_n, J^n)]
= \frac{1}{v_\mathrm{sum}}\limsup_{n \to \infty} \frac{1}{n}\E_{J^n}[T_n(J^n)]
\end{equation}
for any scheduling problem $\mathcal{S}_m$.
Using the same argument, we obtain several propositions for important classes of probability distributions.

\paragraph{The i.i.d. scenario:}
Suppose $\mathcal{P} = \{P_{\mathcal{J}^n}\}$ is a sequence of i.i.d. distributions of $P_\mathcal{J}$.
Then, by the law of large numbers on the random variable $\frac{1}{v_\mathrm{sum}}T(J)$, we have the following corollary.
\begin{Prop}\label{Prop_expectation_iid}
For any i.i.d. scheduling problem, 
\begin{equation*}\label{eq_Thm_expectation_iid}
\lim_{n \to \infty} \frac{1}{n}\E_{J^n}[\SPAN(\phi^\mathrm{opt}_n, J^n)]
= \frac{1}{v_\mathrm{sum}} \E_{J}[T(J)].
\end{equation*}
\end{Prop}

\paragraph{Markov distributions:}
Suppose the distribution $\PRJ$ is Markovian, i.e., 
\begin{equation*}
\PRJ(j^n) = P(j_n|j_{n -1})P(j_{n-1}|j_{n -2})\cdots P(j_2|j_{1})P(j_1).
\end{equation*}
Then, we see $\frac{1}{nv_\mathrm{sum}}T_n(j^n) \overset{a.s.}{\to} \frac{1}{v_\mathrm{sum}}\E_J[T(J)]$ by the ergodic theorem.
Considering the fact that almost sure convergence implies convergence in probability, we then obtain the following corollary:
\begin{Prop}\label{Prop_expectation_Markov}
For any Markov scheduling problem, 
\begin{equation*}\label{eq_Thm_expectation_Markov}
\lim_{n \to \infty} \frac{1}{n}\E_{J^n}[\SPAN(\phi^\mathrm{opt}_n, J^n)]
= \frac{1}{v_\mathrm{sum}} \E_{J}[T(J)].
\end{equation*}
\end{Prop}

\section{Characterization of optimal rate}\label{Sec_Char-of-opt}
This section is devoted to prove the following theorem.
\begin{Thm_non-iid}
For any scheduling problem $\mathcal{S}_m$, $\bar{E}(\mathcal{S}_m) = R(\mathcal{S}_m)$ holds.
\end{Thm_non-iid}
We also give several important applications of Theorem~\ref{Thm_non-iid} in Section~\ref{subsec_app_Char}.

\subsection{Proof of Theorem~\ref{Thm_non-iid}}\label{subsec_proof_Char}
Using Lemma~\ref{Lem_sche_upper-bound} and Lemma~\ref{Lem_sche_lower-bound}, we show Theorem~\ref{Thm_non-iid} as follows.
\begin{proof}[Proof of Theorem~\ref{Thm_non-iid}]
We show $\bar{E}(\mathcal{S}_m) \geq R(\mathcal{S}_m)$ and $\bar{E}(\mathcal{S}_m) \leq R(\mathcal{S}_m)$ separately.
\par
[Proof of $\bar{E}(\mathcal{S}_m) \geq R(\mathcal{S}_m)$]: For any $\gamma > 0$, we show that the rate $\bar{E}(\mathcal{S}_m) + \gamma$ is achievable by constructing an appropriate sequence of scheduling $\{(\phi_n, S_n)\}$.
\par
By the definition of $\BARE$, $\lim_{n \to \infty} P_{\mathcal{J}^n} \{j^n \mid \frac{1}{n \cdot v_\mathrm{sum}}T_n(j^n) > \BARE + \gamma\} = 0$.
Define $S_n := \{j^n \mid \frac{1}{n \cdot v_\mathrm{sum}}T_n(j^n) > \BARE + \gamma\}$ and $\phi_n$ as a scheduler in Lemma~\ref{Lem_sche_upper-bound}, 
we now see $\lim_{n \to \infty} P_{\mathcal{J}^n} \{j^n \mid \frac{1}{n \cdot v_\mathrm{sum}}T_n(j^n) > \BARE + \gamma\} = 0$.
Therefore, the rest is to show 
\begin{equation*}
\limsup_{n \to \infty} \frac{1}{n} \COST(\phi_n, S_n)\leq \BARE + \gamma.
\end{equation*}
\par
First, by Lemma~\ref{Lem_sche_upper-bound}, 
\begin{equation}\label{eq_SPAN_lemma_up}
\SPAN(\phi^\OPT_n, j^n) \leq \frac{1}{v_\mathrm{sum}} T_n(j^n) + \frac{T_\mathrm{max}}{v_\mathrm{min}} 
\end{equation}
holds for any $j^n \in \mathcal{J}^n$.
We also observe that for any $j^n \in \mathcal{J}^n \setminus S_n$, $\frac{1}{v_\mathrm{sum}} T_n(j^n) \leq n (\BARE + \gamma)$ holds by the definition of $S_n$.
Therefore, taking $\max_{j^n \in \mathcal{J}^n \setminus S_n}$ on both sides of Equation~\eqref{eq_SPAN_lemma_up} yields
\begin{equation*}
\COST(\phi_n, S_n) := \max_{j^n \in \mathcal{J}^n \setminus S_n} \SPAN(\phi_n, j^n) 
\leq \max_{j^n \in \mathcal{J}^n \setminus S_n}\frac{1}{v_\mathrm{sum}} T_n(j^n) + \frac{T_\mathrm{max}}{v_\mathrm{min}}
\leq n (\BARE + \gamma) + \frac{T_\mathrm{max}}{v_\mathrm{min}}.
\end{equation*}
Therefore,  taking $\limsup_{n \to \infty}$ yields
\begin{equation*}
\limsup_{n \to \infty} \frac{1}{n}\COST(\phi_n, S_n) \leq \BARE + \gamma
\end{equation*}
which completes proof.
\par
[Proof of $\bar{E}(\mathcal{S}_m) \leq R(\mathcal{S}_m)$]:
To prove this, assume $R(\mathcal{S}_m) < \BARE$ and derive a contradiction.
Let $\gamma > 0$ be a small real number and suppose $\BARE - 2 \gamma$ is achievable. That is, there is a sequence $\{(\phi_n, S_n)\}$ such that
\begin{equation*}
\lim_{n \to \infty} P_{\mathcal{J}^n}(S_n) = 0 \text{~and~}
\limsup_{n \to \infty} \frac{1}{n} \COST(\phi_n, S_n)\leq \BARE -2 \gamma
\end{equation*}
hold. This implies that there is $n_0 \in \mathbb{N}$ such that for any $n \geq n_0$, 
\begin{equation*}
\frac{1}{n} \COST(\phi_n, S_n)\leq \BARE - \gamma
\end{equation*}
holds. Therefore, by the definition of $\COST(\phi_n, S_n)$, we have
\begin{equation*}
\forall n \geq n_0, \quad \forall j^n \in \mathcal{J}^n \setminus S_n, \quad \SPAN(\phi_n, j^n) \leq n (\BARE - \gamma)
\end{equation*}
which implies that, for any $n \geq n_0$ and any $j^n \in \mathcal{J}^n$,
\begin{equation*}
\SPAN(\phi_n, j^n) > n(\BARE - \gamma) \Longrightarrow j^n \in S_n.
\end{equation*}

Together with the monotonicity of the probability measure: $A \subseteq B \Rightarrow P(A) \leq P(B)$, 
we have for any $n \geq n_0$, 
\begin{equation}\label{eq_Thm_Char_1}
P_{\mathcal{J}^n}(S_n) \geq P_{\mathcal{J}^n}(\{j^n \mid \SPAN(\phi_n, j^n) > n(\BARE - \gamma)\}).
\end{equation}
On the other hand, the definition of $\BARE$ ensures that there is $\varepsilon_0 > 0$ such that
\begin{equation}\label{eq_Thm_Char_2}
\exists \{n_i\}_i \subseteq \mathbb{N} \text{~s.t.~} \forall i \in \mathbb{N}, 
\quad P_{\mathcal{J}^{n_i}}(\{j^{n_i} \mid \frac{1}{{n_i} \cdot v_\mathrm{sum}} T_{n_i}(j^{n_i}) > \BARE - \gamma\}) \geq \varepsilon_0.
\end{equation}
Combining Lemma~\ref{Lem_sche_lower-bound} and the two inequalities~(\ref{eq_Thm_Char_1}) and~(\ref{eq_Thm_Char_2}), we obtain
\begin{align}
P_{\mathcal{J}^{n_i}}(S_{n_i})&\geq P_{\mathcal{J}^{n_i}}(\{j^{n_i} \mid \SPAN(\phi_{n_i}, j^{n_i}) > {n_i}(\BARE - \gamma)\})\\
&\geq P_{\mathcal{J}^{n_i}}(\{j^{n_i} \mid \frac{1}{ v_\mathrm{sum} \cdot n_i} T_{n_i}(j^{n_i})> {n_i}(\BARE - \gamma)\})\\
&\geq \varepsilon_0.
\end{align}
This leads to a contradiction because as $i \to \infty$, $P_{\mathcal{J}^{n_i}}(S_{n_i}) \to 0$  holds but $\varepsilon_0$ is a positive constant.
This completes proof.
\end{proof}

\subsection{Applications of Theorem~\ref{Thm_non-iid}}\label{subsec_app_Char}
We now know that the optimal rate $R(\mathcal{S}_m)$ is completely characterized by the quantity $\BARE$.
In this section, we investigate how the quantity $\BARE$ changes depending on scheduling problems $\mathcal{S}_m$.
\paragraph{The i.i.d. scenario:}
Let us analyze the value of $\BARE$ when $\mathcal{P} = \{P_{\mathcal{J}^n}\}$ is a sequence of i.i.d. distributions of $P_\mathcal{J}$.
Recall that $\BARE$ is defined as
\begin{equation*}
\bar{E}(\mathcal{S}_m) := \inf\{\alpha \geq 0 \mid \lim_{n \to \infty} P_{\mathcal{J}^n}\{j^n \mid \frac{1}{n \cdot v_\mathrm{sum}}T_n(j^n) > \alpha\} = 0\}
\end{equation*}
where $T_n(j^n) = \sum_{i \leq n}T(j_i)$.
Then, by the law of large numbers on the random variable $\frac{1}{v_\mathrm{sum}}T(J)$, we have, for any $\varepsilon > 0$, 
\begin{equation}\label{eq_subsec_App_law-of}
\lim_{n \to \infty} \Pr\left\{\left|\frac{1}{m}\E_J[T(J)] - \frac{1}{n \cdot v_\mathrm{sum}}T_n(j^n)\right| > \varepsilon\right\}
= 0
\end{equation}
This shows that $\BARE = \frac{1}{v_\mathrm{sum}}\E_J[T(J)]$ holds, since Equation~(\ref{eq_subsec_App_law-of}) implies 
for any $\BARE + \varepsilon$, $\lim_{n \to \infty} P_{\mathcal{J}^n}\{j^n \mid \frac{1}{n \cdot v_\mathrm{sum}}T_n(j^n) > \BARE + \varepsilon\} = 0$ 
and 
for any $\BARE - \varepsilon$, $\lim_{n \to \infty} P_{\mathcal{J}^n}\{j^n \mid \frac{1}{n \cdot v_\mathrm{sum}}T_n(j^n) > \BARE - \varepsilon\} = 1$.
\par 
Therefore, we have
\begin{Cor}
When $\mathcal{P}$ is i.i.d., $\BARE = \frac{1}{v_\mathrm{sum}}\E_J[T(J)]$ holds.
\end{Cor}
\paragraph{Mixture distributions:}
Next, we analyze $\BARE$ when $P_{\mathcal{J}^n}$ is defined as a convex combination of i.i.d. distributions $\{P^i_{\mathcal{J}^n}\}$.
For simplicity, we assume the distribution $P_{\mathcal{J}^n}$ is expressed by two i.i.d. distributions $P^1_{\mathcal{J}^n}$ and $P^2_{\mathcal{J}^n}$:
\begin{equation*}
P_{\mathcal{J}^n}(j^n) = \alpha_1 P^1_{\mathcal{J}^n}(j^n) + \alpha_2 P^2_{\mathcal{J}^n}(j^n)
\end{equation*}
for $\alpha_1 + \alpha_2 = 1, \alpha_1 > 0, \alpha_2 > 0$.
In this case, we get
\begin{align}
\PRJ\left(\frac{1}{n \cdot v_\mathrm{sum}}T_n(j^n) > \alpha\right)&= \sum_{\substack{j^n \in \mathcal{J}^n\\ (1/n\cdot v_\mathrm{sum})T_n(j^n) > \alpha}} \PRJ(j^n)\\
&= \alpha_1P^1_{\mathcal{J}^n}\left(\frac{1}{n\cdot v_\mathrm{sum}}T_n(j^n) > \alpha\right) + \alpha_2P^2_{\mathcal{J}^n}\left(\frac{1}{n\cdot v_\mathrm{sum}}T_n(j^n) > \alpha\right).
\end{align}
Therefore, applying the law of large numbers to each of $P^1_{\mathcal{J}^n}$ and $P^2_{\mathcal{J}^n}$ yields the following corollary.
\begin{Cor}
For a scheduling problem $\BARE = (\mathcal{J}, T, \mathcal{V}, \mathcal{P})$ whose distribution is a mixture of two i.i.d. distributions $\PRJ^1$ and $\PRJ^2$,
\begin{equation*}
\BARE
= \max\left\{\frac{1}{v_\mathrm{sum}}\E^1_J[T(J)], \frac{1}{v_\mathrm{sum}}\E^2_J[T(J)]\right\}
\end{equation*}
holds where $\E^i_J[T(J)]~(i = 1, 2)$ is the expectation of $T(J)$ when the distribution is $P^i_\mathcal{J}$.
\end{Cor}
In fact, this corollary can be generalized for a mixture of \emph{countably infinitely many,}  \emph{general} (i.e., non-i.i.d.) distributions.
Suppose $\mathcal{P} = \{\PRJ\}$ is defined as
\begin{equation*}
\PRJ(j^n)
= \sum_i \alpha_i \PRJ^i(j^n)
\end{equation*}
where $\{\alpha_i\}_i \subseteq [0, 1]$ satisfies $\sum_i \alpha_i = 1$ and $\PRJ^i$ is a general distribution on $\mathcal{J}^n$.
In this case, we can show the following in a similar manner as~\cite[Theorem~1.4.2]{Han03}.
\begin{Cor}
For a scheduling problem $\BARE = (\mathcal{J}, T, \mathcal{V}, \mathcal{P})$ whose distribution is a mixture of distributions $\{\PRJ^i\}$,
\begin{equation*}
\BARE
= \sup_{i: \alpha_i > 0}\bar{E}(\mathcal{S}^i_m)
\end{equation*}
holds where $\mathcal{S}^i_m := (\mathcal{J}, T, \mathcal{V}, \{\PRJ^i\})$.
\end{Cor}

\paragraph{Markov distributions:}
As mentioned in Proposition~\ref{Prop_expectation_Markov}, $\frac{1}{nv_\mathrm{sum}}T_n(j^n) \overset{a.s.}{\to} \frac{1}{v_\mathrm{sum}}\E_J[T(J)]$ holds by the ergodic theorem.
Considering the fact that almost sure convergence implies convergence in probability, we then obtain the following corollary:
\begin{Cor}
For a scheduling problem $\mathcal{S}_m$ with a Markovian distribution, 
$\BARE =\frac{1}{v_\mathrm{sum}} \E_J[T(J)]$ holds.
\end{Cor}

\section{Second order asymptotics}\label{Sec_Sec-order_asymptotics}
In Section~\ref{Sec_Char-of-opt}, we showed Theorem~\ref{Thm_non-iid} which tells us 
that the optimal scheduling $\{(\phi^\textrm{opt}_n, S^\textrm{opt}_n)\}$ satisfies
\begin{equation}\label{eq_sec-order_1}
\COST(\phi^\textrm{opt}_n, S^\textrm{opt}_n) = \BARE \cdot n + o(n)
\end{equation}
when $\Pr(S_n^\textrm{opt})$ satisfy the condition $\Pr(S_n^\textrm{opt}) \to 0$.
This is further extended to Theorem~\ref{Thm_strong-converse} which shows that for some scheduling problems such as the i.i.d. scenario, 
the optimal rate does not change even if the condition is relaxed to $\Pr(S_n^\textrm{opt}) \leq \varepsilon$ for some $\varepsilon < 1$.
In Section~\ref{Sec_Sec-order_asymptotics}, we refine the equation~\ref{eq_sec-order_1} and show the explicit factors hidden 
in $o(n)$ in the case of the most fundamental scenario, the i.i.d. scheduling problem.

\begin{Thm_Sec}
Let $\mathcal{S}_m$ be an i.i.d. scheduling problem.
Suppose $(\phi_n^\textrm{opt}, S_n^\textrm{opt})$ is an optimal scheduling among all schedulings $(\phi_n, S_n)$ for $\mathcal{S}_m$ satisfying $\Pr(S_n) \leq \varepsilon$. Then, 
\begin{equation*}
\COST\SCHEP
= \frac{n}{v_\mathrm{sum}}\E[T(J)] - \frac{\sqrt{V(\mathcal{S}_m)n}}{v_\mathrm{sum}}\Phi^{-1}(\varepsilon) + \Theta(1).
\end{equation*}
holds where $\Phi(x)$ is the standard Gaussian distribution and $V(\mathcal{S}_m) := \E_J[(T(J) - \E[T(J)])^2]$.
\end{Thm_Sec}

\begin{proof}
First, we observe that the optimal set is expressed as $S^\OPT_n = \{ j^n \mid \frac{1}{n \cdot v_\mathrm{sum}} T_n(j^n) > R_n^+(\varepsilon)\}$
where $R_n^+(\varepsilon) := \inf\{\alpha \geq 0 \mid \Pr(\frac{1}{n \cdot v_\mathrm{sum}} T_n(J^n) > \alpha) \leq \varepsilon\}$.
(We can check that for any scheduling $(\phi_n, S_n)$ with the condition $\Pr(S_n) \leq \varepsilon$, $\COST(\phi_n^\OPT, S^\OPT_n) \leq \COST\SCHE$ holds.)
Note that $\phi_n^\OPT$ is defined in Lemma~\ref{Lem_sche_upper-bound}.
Therefore, the optimal scheduling is a pair $(\phi^\OPT_n, S^\OPT_n)$.
\par
We now evaluate $\frac{1}{n}\COSTOPT$. By Lemma~\ref{Lem_sche_upper-bound} and Lemma~\ref{Lem_sche_lower-bound}, 
\begin{equation}\label{eq_Thm_ineq_1}
\frac{1}{n\cdot v_\mathrm{sum}}\max_{j^n \in S_n^c}T_n(j^n) \leq \frac{1}{n} \max_{j^n \in S_n^c} \SPAN\SCHEP = \frac{1}{n}\COSTOPT 
\leq \frac{1}{n\cdot v_\mathrm{sum}}\max_{j^n \in S_n^c}T_n(j^n) + \frac{T_\mathrm{max}}{n\cdot v_\mathrm{min}}
\end{equation}
holds. We also obtain $\frac{1}{n\cdot v_\mathrm{sum}}\max_{j^n \in S_n^c}T_n(j^n) = R_n^+(\varepsilon)$ by the definition of $R_n^+(\varepsilon)$. Therefore, the equation~(\ref{eq_Thm_ineq_1}) is simplified to
\begin{equation}\label{eq_Thm_ineq_2}
 R_n^+(\varepsilon)\leq  \frac{1}{n}\COSTOPT  \leq   R_n^+(\varepsilon)+ \frac{T_\mathrm{max}}{n\cdot v_\mathrm{min}}.
\end{equation}
This means that it is sufficient to analyze the term $R^+_n(\varepsilon)$.
\par
We now apply the Berry-Esseen theorem:
\begin{Thm*}[Berry-Esseen]
Let $X_i$ be an i.i.d. random variable with zero mean and $T := \E[|X_i|^3] < \infty$. Then
\begin{equation}\label{eq_Berry-Esseen}
\sup_{a \in \mathbb{R}} \left| \Pr\left(\frac{1}{\sigma \sqrt{n}}\sum_{i\leq n}X_i < a\right) - \Phi(a)\right| \leq \frac{T}{\sigma^3\sqrt{n}}
\end{equation}
where $\sigma^2 := \E[X_i^2]$.
\end{Thm*}
For any $R \in \mathbb{R}$, applying the Berry-Esseen theorem with 
\[
X_i := \frac{1}{v_\mathrm{sum}}(T[J_i] - \E[J]) \text{~and~} a = n\cdot\frac{v_\mathrm{sum}\cdot R - \E[T(J)]}{\sqrt{V(\mathcal{S}_m)n}}
\]
yields
\begin{equation*}
 \left| \Pr\left(\frac{1}{n\cdot v_\mathrm{sum}} T_n(J^n)< R\right) - \Phi\left(n\cdot\frac{v_\mathrm{sum}\cdot R - \E[T(J)]}{\sqrt{V(\mathcal{S}_m)n}}\right)\right| 
 \leq \frac{T}{\sqrt{nV(\mathcal{S}_m)^3}}.
\end{equation*}
This implies
\begin{align*}
\Pr\left(\frac{1}{n\cdot v_\mathrm{sum}} T_n(J^n) \geq R\right) - \frac{T}{\sqrt{nV(\mathcal{S}_m)^3}}
&\leq 1 - \Phi\left(n\cdot\frac{v_\mathrm{sum}\cdot R - \E[T(J)]}{\sqrt{V(\mathcal{S}_m)n}}\right)\\
&= \Phi\left(-n\cdot\frac{v_\mathrm{sum}\cdot R - \E[T(J)]}{\sqrt{V(\mathcal{S}_m)n}}\right)\\
& \leq \Pr\left(\frac{1}{n\cdot v_\mathrm{sum}} T_n(J^n) \geq R\right) + \frac{T}{\sqrt{nV(\mathcal{S}_m)^3}}
\end{align*}
where $\Phi(x) = 1 - \Phi(-x)$ is used.
We therefore obtain, by the simple relation:
\[
\Pr\left(\frac{1}{n\cdot v_\mathrm{sum}} T_n(J^n) > R\right)  \leq \Pr\left(\frac{1}{n\cdot v_\mathrm{sum}} T_n(J^n) \geq R\right),
\]
\begin{equation*}
\Pr\left(\frac{1}{n\cdot v_\mathrm{sum}} T_n(J^n) > R\right) - \frac{T}{\sqrt{nV(\mathcal{S}_m)^3}}
\leq \Phi\left(-n\cdot\frac{v_\mathrm{sum}\cdot R - \E[T(J)]}{\sqrt{V(\mathcal{S}_m)n}}\right)
\end{equation*}
and, by taking the sequence $R_n \searrow R$ and the continuity of $\Phi(\cdot)$ and the probability measure,
\begin{equation*}
\Phi\left(-n\cdot\frac{v_\mathrm{sum}\cdot R - \E[T(J)]}{\sqrt{V(\mathcal{S}_m)n}}\right)
 \leq \Pr\left(\frac{1}{n\cdot v_\mathrm{sum}} T_n(J^n) > R\right) + \frac{T}{\sqrt{nV(\mathcal{S}_m)^3}}
\end{equation*}
for any $R \in \mathbb{R}$.
By substituting $R = R_n^+(\varepsilon)$ and taking the inverse $\Phi^{-1}$,  these inequalities imply
\begin{equation}\label{eq_Thm_Sec_order_1}
\Phi^{-1}\left(\varepsilon - \frac{T}{\sqrt{nV(\mathcal{S}_m)^3}}\right)
\leq -n\cdot\frac{v_\mathrm{sum}\cdot R - \E[T(J)]}{\sqrt{V(\mathcal{S}_m)n}}
\leq \Phi^{-1}\left(\varepsilon + \frac{T}{\sqrt{nV(\mathcal{S}_m)^3}}\right).
\end{equation}
We now expand  the term
$\Phi^{-1}\left(\varepsilon + \frac{T}{\sqrt{nV(\mathcal{S}_m)^3}}\right)$
by the Taylor approximation as follows
\begin{equation}\label{eq_Thm_Sec_order_2}
\Phi^{-1}\left(\varepsilon + \frac{T}{\sqrt{nV(\mathcal{S}_m)^3}}\right)
= \Phi^{-1}(\varepsilon) + \Theta \left(\frac{1}{\sqrt{n}}\right).
\end{equation}
Combining the inequality~(\ref{eq_Thm_Sec_order_1}) and the equation~(\ref{eq_Thm_Sec_order_2}), we get
\begin{equation*}
n v_\mathrm{sum} \cdot R_n^+(\varepsilon)
= n\E[T(J)] - \sqrt{V(\mathcal{S}_m)n}\Phi^{-1}(\varepsilon) + \Theta(1).
\end{equation*}
Together with the inequality~(\ref{eq_Thm_ineq_2}), we finally obtain
\begin{equation*}
\COST\SCHEP
= \frac{n}{v_\mathrm{sum}}\E[T(J)] - \frac{\sqrt{V(\mathcal{S}_m)n}}{v_\mathrm{sum}}\Phi^{-1}(\varepsilon) + \Theta(1).
\end{equation*}
This completes proof.
\end{proof}

\section*{Acknowledgement}
DS would like to take this opportunity to thank the "Nagoya University Interdisciplinary Frontier Fellowship" supported by JST and Nagoya University.
This work was supported by the MEXT Quantum Leap Flagship Program (MEXT Q-LEAP) grant No. JPMXS0120319794.

\bibliography{Citations/Scheduling, Citations/quant_info_D, Citations/comm_comp_D, Citations/books_D}

\bibliographystyle{unsrt}
\appendix
\section{Strong converse theorem}\label{Sec_converse}
Theorem~\ref{Thm_non-iid} tells us that the optimal rate of scheduling problem $\mathcal{S}_m$ is characterized by the quantity $\BARE$
when schedulings need to satisfy $P(S_n) \to 0$. 
Interestingly, in some scheduling problems, the probability $P(S_n)$ must approach to $1$ 
when a scheduling $\{(\phi_n, S_n)\}$ is designed to satisfy $\lim \frac{1}{n}\COST(\phi_n, S_n) \leq R(S_m) - \varepsilon$ for any small $\varepsilon$.
This is called the strong converse property in the information theory literature.
The definition of the strong converse property is as follows:
\begin{Def}
A scheduling problem $\mathcal{S}_m$ satisfies the strong converse property if and only if 
for any $\varepsilon >0$ and any scheduling $\{\phi_n, S_n\}$ for $\mathcal{S}_m$, 
$\limsup_{n \to \infty} \frac{1}{n}\COST(\phi_n, S_n) \leq R(S_m) - \varepsilon$ implies $\lim P(S_n) = 1$.
\end{Def}

In Theorem~\ref{Thm_strong-converse} below, we give a necessary and sufficient condition for the strong converse property.

\begin{Thm}\label{Thm_strong-converse}
A scheduling problem $\mathcal{S}_m$ satisfies the strong converse property if and only if $\BARE = \UBARE$ holds
where
\begin{equation*}
\UBARE := \sup\left\{\beta \geq 0 \mid \lim_{n \to \infty} P_{\mathcal{J}^n}\{j^n \mid \frac{1}{n \cdot v_\mathrm{sum}}T_n(j^n) < \beta\} = 0\right\}.
\end{equation*}
\end{Thm}
\begin{Rem}
By definition, $\UBARE \leq \BARE$ holds. Using the law of large numbers and the ergodic theorem respectively, 
we see that the i.i.d. scenario and the markov distribution in Section~\ref{subsec_app_Char} satisfy the condition $\BARE = \UBARE$.
On the other hand, in the case of mixture distributions, the condition does not hold in general.
\end{Rem}

\subsection{Proof of Theorem~\ref{Thm_strong-converse}}
The proof of Theorem~\ref{Thm_strong-converse} is as follows.
\begin{proof}[proof of Theorem~\ref{Thm_strong-converse}]
The proof consists of the sufficiency part and the necessity part.
\par
\paragraph{[Sufficiency]:} Assume $\BARE = \UBARE$.
Suppose a scheduling $\{(\phi_n, S_n)\}$ satisfies $\lim \frac{1}{n}\COST(\phi_n, S_n) \leq R(S_m) - 2\varepsilon = E(S_m) - 2\varepsilon$.
This implies that there is $n_0$ such that for any $n > n_0$, $\frac{1}{n}\COST(\phi_n, S_n) < \BARE - \varepsilon$ holds.
By the definition of $\COST(\phi_n, S_n)$, this further implies that for any $j^n \in \mathcal{J}^n\setminus S_n$, 
\begin{equation*}
\frac{1}{n\cdot v_\mathrm{sum}} T_n(j^n) \leq \COST \SCHE < \BARE - \varepsilon
\end{equation*}
holds. Therefore, we see 
\begin{align*}
\Pr(S_n^c) &\leq \Pr(\{j^n| \frac{1}{n\cdot v_\mathrm{sum}}T_n(j^n) < \BARE - \varepsilon\})\\
&= \Pr(\{j^n| \frac{1}{n\cdot v_\mathrm{sum}}T_n(j^n) < \UBARE - \varepsilon\})\\
&\to 0
\end{align*}
where the first equality follows from the assumption $\BARE = \UBARE$ and the convergence comes from the definition of $\UBARE$.
This shows $\Pr(S_n) \to 1$ and completes the proof.
\par
\paragraph{[Necessity]:} Assume that for any $\varepsilon > 0$ and any scheduling $\{\SCHE\}$ satisfying $\limsup_{n \to \infty} \frac{1}{n}\COST(\phi_n, S_n) \leq R(S_m) - \varepsilon$,
$\Pr(S_n) \to 1$ holds. Define a scheduler $\phi_n$ as in Lemma~\ref{Lem_sche_upper-bound} and define $S_n := \{j^n \mid \frac{1}{n\cdot v_\mathrm{sum}}T_n(j^n) \geq \BARE - \varepsilon\}$.
We first show the scheduler satisfies $\lim \frac{1}{n}\COST(\phi_n, S_n) \leq \BARE - \varepsilon$.
\par
By Lemma~\ref{Lem_sche_upper-bound}, we see $\SPAN(\phi_n, j^n) \leq \frac{1}{v_\mathrm{sum}}T_n(j^n) + \frac{T_\mathrm{max}}{v_\mathrm{min}}$ and therefore we have
\begin{equation*}
\frac{1}{n}\COST(\phi_n, S_n)
= \frac{1}{n} \max_{j^n \in S_n^c} \SPAN(\phi_n, j^n) \leq \max_{j^n \in S_n^c}\frac{1}{n\cdot v_\mathrm{sum}}T_n(j^n) + \frac{1}{n\cdot v_\mathrm{min}}T_\mathrm{max}.
\end{equation*}
Therefore, we obtain
\begin{equation*}
\frac{1}{n}\COST(\phi_n, S_n)
 \leq \max_{j^n \in S_n^c}\frac{1}{n\cdot v_\mathrm{sum}}T_n(j^n) + \frac{1}{n\cdot v_\mathrm{min}}T_\mathrm{max}
 \leq \BARE - \varepsilon + \frac{1}{n\cdot v_\mathrm{min}}T_\mathrm{max}
\end{equation*}
by the definition of $S_n$. Then, taking $\limsup_{n \to \infty}$ of both sides yields
\begin{equation*}
\limsup_{n \to \infty} \frac{1}{n}\COST(\phi_n, S_n)
 \leq \BARE - \varepsilon.
\end{equation*}

This in turn implies $\Pr(S_n) \to 1$ and equivalently, $\Pr(S^c_n) \to 0$ because of the assumption. Recall that $\UBARE$ is defined as
\begin{equation*}
\UBARE = \sup\left\{\beta \geq 0 \mid \lim_{n \to \infty} P_{\mathcal{J}^n}\{j^n \mid \frac{1}{n \cdot v_\mathrm{sum}}T_n(j^n) < \beta\} = 0\right\}.
\end{equation*}
Therefore, $\Pr(S^c_n) = \Pr(\{j^n \mid \frac{1}{n\cdot v_\mathrm{sum}}T_n(j^n) < \BARE - \varepsilon\}) \to 0$ means $\BARE - \varepsilon \leq \UBARE$ for any $\varepsilon > 0$.
This means $\BARE  \leq \UBARE$ and therefore we obtain $\BARE = \UBARE$ which completes proof.
\end{proof}

\end{document}